\documentclass[11pt]{article}
\usepackage{amsmath,amssymb,amsthm,mathtools,amscd}
 \usepackage{stmaryrd} 
\usepackage{hyperref}
\usepackage{geometry}
\title{Fungi as functors:\\ A category-theoretic approach to mycelial organisation}

\author{Andrew Adamatzky\\Unconventiional Computing Lab, UWE Bristol, UK}

\date{}

\theoremstyle{definition}
\newtheorem{definition}{Definition}[section]
\newtheorem{assumption}[definition]{Assumption}
\theoremstyle{plain}
\newtheorem{theorem}[definition]{Theorem}
\newtheorem{proposition}[definition]{Proposition}

\theoremstyle{remark}
\newtheorem{remark}[definition]{Remark}

\newcommand{\Obj}{\mathrm{Ob}}
\newcommand{\Hom}{\mathrm{Hom}}
\newcommand{\id}{\mathrm{id}}
\newcommand{\Env}{\mathbf{Env}}
\newcommand{\Myc}{\mathbf{Myc}}
\newcommand{\Prog}{\mathbf{Prog}}
\newcommand{\Fld}{\mathbf{Fld}}

\begin{document}
\maketitle

\begin{abstract}
We develop a rigorous, equation-free category-theoretic foundation for fungal organisation. A fungal organism is formalised as a functor from a category $\Env$ of structured environmental states and admissible transformations to a category $\Myc$ of mycelial network states and biologically meaningful morphisms. An operational program category $\Prog$ models time-ordered exposure protocols, and a semantics functor $\mathcal{F}_{\mathrm{prog}}:\Prog\to\Myc$ maps experimental perturbations to induced network transformations. Species and strain variability are expressed as natural transformations between fungal functors, and ecological feedback is captured via an adjunction between sensing and environment modification. Network fusion (anastomosis) is identified with pushouts in $\Myc$, and order effects in exposure sequences are quantified by a local Lie structure and a Baker--Campbell--Hausdorff expansion near the identity program. A minimal worked exposure example demonstrates how non-commutativity yields experimentally testable quadratic scaling of order asymmetry. The framework provides a structurally explicit and falsifiable basis for analysing compositional perturbations, mixture coupling, robustness limits, and cross-species comparability in fungal systems.\\

\end{abstract}

\section{Introduction}

Fungal mycelia are adaptive networked organisms: they alter morphology (branching, fusion, thickening), redistribute resources, and modulate electrical activity in response to environmental perturbations. We propose a structural, category-theoretic foundation in which a fungal organism is modelled as a functor from environments (and programs acting on environments) to mycelial states and transformations. The goal is not to substitute biological detail with abstraction, but to formalise what is \emph{invariant} across substrates, parameterisations, and measurement modalities: compositionality of perturbations, compositionality of adaptation, and the structural constraints that relate them.

\subsection*{Why categories for mycelia?}

Mycelia are inherently \emph{compositional} objects: they are networks that
grow by local operations (branching, fusion, thickening), and they are probed
by experimentally composable perturbations (pulses, ramps, mixtures, spatial
rearrangements). Category theory is a natural language for such systems because
it makes compositional structure explicit: transformations are first-class
entities (morphisms), multi-step interventions are represented by composition,
and structural ``gluing'' operations are captured by universal constructions
such as pushouts and pullbacks. In this setting, the point is not to claim that
fungal dynamics \emph{are} categories, but that many experimentally relevant
invariants---order dependence, interface composition, cross-condition
comparability, and robust coarse-graining---are most naturally formulated at
the level of morphisms and universal properties rather than at the level of a
single mechanistic equation set.

A second motivation is methodological: fungal systems are studied across
modalities (imaging, electrophysiology, chemical sensing, mechanical response),
each producing its own state representation and its own notion of ``the same''
perturbation. Functorial semantics separates \emph{syntax} (what we do:
programs in the laboratory) from \emph{semantics} (what changes in the organism:
induced network morphisms), and makes explicit when different experimental
pipelines define compatible abstractions. This perspective echoes the broader
applied-category-theory program in which compositional reasoning is used to
connect heterogeneous models and measurements without forcing premature
reduction to a single dynamical formalism \cite{fong2018seven,spivak2011category,agmon2024foundations,aduddell2024compositional}.

\subsection*{Category theory in biology: context and precedent}

Category-theoretic ideas have a long history in theoretical biology and the
study of natural systems, where the emphasis is on relations between
descriptions, observables, and dynamical organisation rather than on any one
preferred set of equations \cite{louie1983categorical}. More recently, category
theory has been used to formalise compositional structure in biological and
bio-inspired settings, including semantic/relational modelling of biological
materials via ologs \cite{spivak2011category}, compositional approaches to systems
biology \cite{agmon2024foundations,aduddell2024compositional}, and category-theoretic
frameworks proposed in neuroscience to relate structural descriptions across
scales \cite{northoff2019mathematics}. These precedents support the premise
that category theory can serve as a \emph{unifying metalanguage} for biological
organisation: it provides a principled way to state when two experimental
manipulations are ``the same'' at an abstract level, how complex interventions
are assembled from simpler ones, and how multiscale descriptions relate.

\subsection*{Fungal signalling as a compositional target}

Fungal electrophysiology and fungal bioelectronics have rapidly developed,
including increasingly careful experimental protocols and analyses suggesting
structured electrical dynamics correlated with growth and environmental
adaptation. Fluctuations in electrical potential and spike-like events have been observed in a wide range of
fungal species and substrates~\cite{slayman1976action,olsson1995action,adamatzky2018spiking,adamatzky2022language,dehshibi2021electrical,fukasawa2024electrical,buffi2025electrical,fukasawa2025electrical}. Such activity has been associated with physiological processes including growth, nutrient transport, environmental sensing, and responses to mechanical or chemical perturbations~\cite{gow1984transhyphal,harold1985fungi,gow1989relationship,gow1995electric,feng2019analysis,phillips2023electrical,fukasawa2023electrical}. While the existence of fungal electrical signalling is now well
established, most studies rely on recordings from single electrodes or linear electrode arrangements, limiting insight into the spatial organisation of electrical dynamics.
In particular, recent work reports multi-site electrical activity and its modulation under
controlled conditions, raising the question of whether observed events are
independent local responses or manifestations of propagating and interacting
signals on a network substrate~\cite{fukasawa2023electrical,adamatzky2026propagation,adamatzky2026directional}. The present framework is intended to complement
such experimental progress: it provides an equation-free way to express and
test compositional claims (e.g.\ order effects of exposures, mixture coupling,
and fusion as gluing), and to compare these claims across strains, substrates,
and measurement pipelines.

\section{Foundational Categories}

\subsection{Category of Environments: $\Env$}

We model environments as structured states carrying spatial and field information as well as constraints. The level of structure can be chosen to match an experiment (e.g.\ eco-chamber) or an ecological situation.

\begin{definition}[Environmental object]
An \emph{environmental object} is a quadruple
\[
E=(G,\rho,\phi,\chi)
\]
where:
\begin{itemize}
\item $G$ is a finite (or locally finite) graph encoding substrate topology (e.g.\ pores, channels, contact regions), or an embedded graph approximating a spatial domain.
\item $\rho:V(G)\to \mathbb{R}_{\ge 0}$ is a nutrient/resource field on vertices (extensions to edges or continuous domains are admissible).
\item $\phi:V(G)\to \mathbb{R}^k$ is a chemical/VOC field (possibly vector-valued, $k\ge 1$).
\item $\chi$ is a constraint datum (e.g.\ humidity/temperature bounds, mechanical constraints, boundary conditions, permissible airflow), modelled as an element of a specified constraint space $\mathcal{C}(G)$.
\end{itemize}
Denote the class of such objects by $\Obj(\Env)$.
\end{definition}

We now define morphisms as admissible environmental transformations (redistribution, injection, deformation, reparameterisation), required only to respect declared constraints.

\begin{definition}[Environmental morphism]
Let $E_1=(G_1,\rho_1,\phi_1,\chi_1)$ and $E_2=(G_2,\rho_2,\phi_2,\chi_2)$. An \emph{environmental morphism}
\[
f:E_1\to E_2
\]
is a tuple $f=(f_G,f_\rho,f_\phi,f_\chi)$ where:
\begin{itemize}
\item $f_G : G_1 \to G_2$ is a graph morphism in a fixed chosen
category $\mathbf{Graph}$ of finite (or locally finite) graphs.
Throughout this paper, we assume $\mathbf{Graph}$ is a category
whose morphisms are structure-preserving maps closed under composition
and containing all identity maps. In particular, we assume that
$\mathbf{Graph}$ admits pushouts along monomorphisms and that
monomorphisms are stable under pullback.
\item $f_\rho$ transforms resources compatibly with $f_G$ (e.g.\ pushforward/pullback rule) and obeys any conservation or budget constraints specified by $\chi_1,\chi_2$.
\item $f_\phi$ transforms chemical fields compatibly and obeys injection/dilution bounds specified by $\chi_1,\chi_2$.
\item $f_\chi$ maps constraints in a way that preserves admissibility (e.g.\ allowable humidity ranges, chamber limits).
\end{itemize}
Composition is defined componentwise, with induced compositions for the field transforms.
The identity morphism on $E$ is $(\id_G,\id_\rho,\id_\phi,\id_\chi)$.
\end{definition}

\begin{proposition}
The class of environmental objects and morphisms
defined above forms a category $\mathbf{Env}$.
\end{proposition}
\begin{proof}
Identities are given componentwise by
$(\mathrm{id}_G,\mathrm{id}_\rho,\mathrm{id}_\phi,\mathrm{id}_\chi)$.
Composition is defined componentwise using composition in
$\mathbf{Graph}$ together with compatible composition rules
for resource, chemical, and constraint transforms.
Associativity follows from associativity in $\mathbf{Graph}$
and closure of admissible field and constraint transforms
under composition. \end{proof}

\subsection*{Biological Interpretation}

An environmental object $E=(G,\rho,\phi,\chi)$ corresponds,
experimentally, to a controlled substrate configuration.
For example:

\begin{itemize}
\item $G$ may represent pore networks in a colonised block,
contact topology in a Petri dish, or chamber connectivity
in an eco-chamber experiment.
\item $\rho$ models spatial nutrient gradients or depletion zones.
\item $\phi$ models volatile organic compound (VOC) fields
or chemical signalling gradients.
\item $\chi$ encodes imposed constraints such as humidity,
temperature bounds, airflow regime, or mechanical confinement.
\end{itemize}

An environmental morphism represents a controlled perturbation:
a VOC pulse, nutrient redistribution, humidity ramp,
or geometric deformation of substrate.
The categorical structure formalises the compositionality
of such perturbations without assuming any specific
biophysical transport model.

\subsection{Category of Mycelial States: $\Myc$}

We represent a fungal organism (at a chosen observation scale) as a weighted network with dynamic/electrical state.

\begin{definition}[Mycelial object]
A \emph{mycelial state} is a triple
\[
M=(T,\sigma,\omega)
\]
where:
\begin{itemize}
\item $T$ is a finite connected graph (often embedded in the substrate graph $G$ or in space), representing hyphal network topology.
\item $\sigma:E(T)\to \mathbb{R}_{\ge 0}$ assigns edge weights (conductivity, thickness, transport capacity).
\item $\omega:V(T)\to \Omega$ assigns node states (electrophysiology; e.g.\ $\Omega=\mathbb{R}$ for potential, or a space of time-series features).
\end{itemize}
Denote $\Obj(\Myc)$ the class of such objects.
\end{definition}

Morphisms represent biologically meaningful network updates (growth, pruning, fusion, reweighting, state evolution).

\begin{definition}[Mycelial morphism]
Let $M_1=(T_1,\sigma_1,\omega_1)$ and $M_2=(T_2,\sigma_2,\omega_2)$. A \emph{mycelial morphism}
\[
g:M_1\to M_2
\]
is a triple $g=(g_T,g_\sigma,g_\omega)$ where:
\begin{itemize}
\item $g_T : T_1 \to T_2$ is a morphism in the fixed graph
category $\mathbf{Graph}$ used in $\Env$.

\item $g_\sigma$ is a weight-transport/update rule producing $\sigma_2$
from $\sigma_1$ and $g_T$. Concretely, we require a function
\[
g_\sigma:\sigma_1 \mapsto \sigma_2
\]
such that $\sigma_2$ is compatible with the identification/coarse-graining
induced by $g_T$ (e.g.\ pushforward under $g_T$ on edges where defined,
and prescribed merging rules when edges are identified).

\item $g_\omega$ maps node states compatibly with $g_T$
(e.g.\ transport of state, coarse-graining, or update operator).
\end{itemize}
Composition is componentwise; identities are identity maps.
\end{definition}

\begin{proposition}
The class of mycelial objects and morphisms
forms a category $\mathbf{Myc}$. 
\end{proposition}
\begin{proof}

Identities are given componentwise.
Composition is defined by composing graph morphisms
in $\mathbf{Graph}$ together with compatible composition rules
for weight maps $\sigma$ and state maps $\omega$.
Associativity and identity laws follow from the corresponding
properties in $\mathbf{Graph}$ and closure of the
state/weight update rules under composition.

\end{proof}

\subsection*{Biological Interpretation}

A mycelial object $(T,\sigma,\omega)$ corresponds to
an experimentally observable fungal state:

\begin{itemize}
\item $T$ represents inferred hyphal network topology
(e.g.\ via imaging and skeletonisation).
\item $\sigma$ encodes measurable transport or conductivity
proxies (e.g.\ cord thickness, inferred hydraulic capacity,
or electrical conductance).
\item $\omega$ represents dynamic or electrophysiological
state, such as extracellular potential or derived
time-series features.
\end{itemize}

A mycelial morphism models biologically meaningful change:
growth, pruning, cord thickening, fusion (anastomosis),
redistribution of transport capacity, or evolution of
electrical state.

\section{The Fungal Functor $\mathcal{F}:\Env\to\Myc$}

We now formalise the statement ``fungi are functors'' by specifying a structure-preserving mapping from environments and their transformations to mycelial states and their transformations.

\begin{definition}[Fungal functor]
A \emph{fungal functor} is a functor
\[
\mathcal{F}:\Env\to\Myc
\]
such that:
\begin{itemize}
\item for each environmental object $E$, $\mathcal{F}(E)$ is the induced mycelial state under $E$;
\item for each environmental morphism $f:E_1\to E_2$, $\mathcal{F}(f):\mathcal{F}(E_1)\to \mathcal{F}(E_2)$ is the induced adaptive transformation.
\end{itemize}
\end{definition}

\begin{remark}
The functorial identity
\[
\mathcal{F}(g \circ f) = \mathcal{F}(g)\circ \mathcal{F}(f)
\]
expresses strict compositionality of environmental morphisms
and induced mycelial transformations. No assumption of
commutativity or linearity is implied.
\end{remark}

\subsection*{Biological Interpretation}

The functor $\mathcal{F}$ formalises the statement that
fungal adaptation is compositional: applying environmental
perturbation $f$ followed by $g$ induces the same adaptive
transformation as applying their composite $g\circ f$.
This captures an experimentally testable invariance:
sequential perturbations should compose at the level
of induced network transformation.

\section{Enrichment and Metric Structure}

Both $\mathbf{Env}$ and $\mathbf{Myc}$ may be enriched
over a monoidal category of metric spaces in the sense of
Lawvere enrichment. Concretely, suppose each hom-set
is equipped with a distance function compatible with
composition.

\begin{definition}[Non-expanding fungal functor]
A fungal functor $\mathcal{F} : \Env \to \Myc$
is \emph{non-expanding} if there exists $L \ge 0$ such that
for all objects $E_1,E_2$,
\[
d_{\Myc}\big(\mathcal{F}(E_1),\mathcal{F}(E_2)\big)
\le L\, d_{\Env}(E_1,E_2).
\]
\end{definition}

This enriched viewpoint allows quantitative comparison
of environmental variation and induced mycelial variation,
but is not required for the core categorical constructions.

\section{Natural Transformations and Species/Strain Variability}

Different species or strains correspond to different functors from the same environment category.

\begin{definition}[Species functors]
Let $\mathcal{F}_1,\mathcal{F}_2:\Env\to\Myc$ be two fungal functors (e.g.\ two species/strains).
A \emph{natural transformation} $\eta:\mathcal{F}_1\Rightarrow \mathcal{F}_2$ is a family of morphisms
\[
\eta_E:\mathcal{F}_1(E)\to \mathcal{F}_2(E)
\]
such that for all $f:E\to E'$,
\[
\eta_{E'}\circ \mathcal{F}_1(f) = \mathcal{F}_2(f)\circ \eta_E.
\]
\end{definition}

\begin{remark}
Naturality expresses a strong comparability claim: changing the environment then translating species is equivalent to translating species then changing the environment. Empirically, failure of commutativity in this square quantifies species-specific sensitivity to environmental transformations.
\end{remark}

\subsection*{Biological Interpretation}

If naturality fails for a given pair of species,
the commutativity square provides a quantitative measure
of species-specific sensitivity to environmental ordering.
Deviation from commutativity can therefore be used to
classify strains by structural response properties
rather than by parameter values of specific models.

\section{Adjunction and Ecological Feedback}

Fungi modify their environment (resource redistribution, substrate modification). We formalise this via a functor from mycelial states to environments.

\begin{definition}[Environment modification functor]
A functor $\mathcal{G}:\Myc\to\Env$ maps a mycelial state to its induced environmental modification (e.g.\ redistribution of $\rho$, changes to constraints $\chi$).
\end{definition}

\begin{definition}[Adjunction]
We say $\mathcal{F}:\Env\to\Myc$ is left adjoint to $\mathcal{G}:\Myc\to\Env$, written $\mathcal{F}\dashv \mathcal{G}$, if there is a natural isomorphism
\[
\Hom_{\Myc}(\mathcal{F}(E),M)\cong \Hom_{\Env}(E,\mathcal{G}(M))
\]
natural in $E\in\Env$ and $M\in\Myc$.
\end{definition}

\begin{remark}
Adjunction captures a duality: \emph{sensing/adaptation} (environment $\to$ mycelium) and \emph{environment modification} (mycelium $\to$ environment). This is a formal expression of ecological feedback.
\end{remark}

\subsection*{Biological Interpretation}

The adjunction $\mathcal{F}\dashv\mathcal{G}$ captures
bidirectional coupling between fungus and environment.
The unit of the adjunction corresponds to sensing/adaptation,
while the counit corresponds to environmental modification
through metabolism, resource depletion, or structural
alteration of substrate.
This expresses ecological feedback without specifying
mechanistic equations.

\section{Limits and Colimits in $\Myc$}

Network phenomena (fusion/anastomosis, merging) correspond naturally to categorical colimits.

\begin{definition}[Pushout as fusion (anastomosis)]
Given morphisms $A\to B$ and $A\to C$ in $\Myc$, a pushout $B\sqcup_A C$ represents fusion of networks $B$ and $C$ along a shared substructure $A$.
\end{definition}

\begin{proposition}[Existence of pushouts along monomorphisms]
Suppose $\mathbf{Graph}$ admits pushouts along monomorphisms
and that weight and state update rules are compatible with
graph identification. Then $\mathbf{Myc}$ admits pushouts
for cospans of the form
\[
A \xrightarrow{m_1} B, \qquad
A \xrightarrow{m_2} C
\]
where $m_1,m_2$ are monomorphisms.
\end{proposition}
\begin{proof}
Construct the pushout graph in $\mathbf{Graph}$ by forming
the quotient of the disjoint union of $B$ and $C$ identifying
images of $A$. Define the induced edge-weight and node-state
updates by the prescribed compatibility rules under identification.
The universal property follows from the universal property in
$\mathbf{Graph}$ together with functorial compatibility of the
weight and state assignments.
\end{proof}

\begin{remark}
Analogously, pullbacks can encode ``shared constraint'' intersections (e.g.\ competition for a shared resource region), while colimits encode network merging and limits encode convergence of constraints.
\end{remark}

\subsection*{Example (Anastomosis as pushout)}

Consider two growing mycelial networks $B$ and $C$ extending
towards a shared substrate region represented by a subnetwork
$A$. Suppose there are monomorphisms
\[
A \xrightarrow{m_1} B,
\qquad
A \xrightarrow{m_2} C
\]
identifying the prospective contact region in each network.

If physical contact occurs and hyphal fusion (anastomosis)
is established, the resulting network is represented by
the pushout
\[
B \sqcup_A C
\]
in $\Myc$.

Categorically, the pushout identifies the images of $A$
inside $B$ and $C$ and merges edge weights and node states
according to the prescribed compatibility rules.
Biologically, this corresponds to the creation of
conductive continuity between two previously distinct
transport systems, allowing redistribution of resources
and electrical signalling across the fused structure.

Thus anastomosis is not merely a biological event but
a canonical colimit construction in the category $\Myc$.

\section{Programs, Semantics, and Falsifiability}\label{sec:programs}

The environment-category view above is static-to-static. Experiments are operational: one applies \emph{programs} (VOC pulses, mixtures, humidity ramps) and measures adaptive transformations. We therefore define a program category and show how it induces a functor into $\Myc$.

\subsection{Program category $\Prog$}

Let $\mathcal{U}$ be a set of admissible control actions (VOC injection profiles, humidity modulation, etc.). A program is $(u,t)$ where $u:[0,t]\to\mathcal{U}$ is measurable and $t\ge 0$.

\begin{definition}[Program concatenation]
For programs $(u,t)$ and $(v,s)$ define concatenation $(v,s)\star(u,t) := (v\star u,t+s)$ where
\[
(v\star u)(\tau)=
\begin{cases}
u(\tau), & 0\le \tau\le t,\\
v(\tau-t), & t<\tau\le t+s.
\end{cases}
\]
\end{definition}

\begin{assumption}[Deterministic state transition system]\label{ass:deterministic}
There exists a state space $\mathcal{S}$ of internal fungal states and a deterministic evolution map
\[
\Phi:\mathcal{S}\times \{(u,t)\}\to \mathcal{S},\qquad (S,(u,t))\mapsto \Phi_{(u,t)}(S),
\]
such that $\Phi$ respects concatenation (causality):
\[
\Phi_{(v,s)\star(u,t)}(S)=\Phi_{(v,s)}(\Phi_{(u,t)}(S)).
\]
\end{assumption}

\begin{definition}[Program category $\Prog$]
Objects are internal states $S\in\mathcal{S}$. A morphism $S\to S'$ is an equivalence class $[u,t]_S$ such that $\Phi_{(u,t)}(S)=S'$, where two programs $(u,t)$ and $(u',t')$ are equivalent at $S$
if and only if
\[
\Phi_{(u,t)}(S)=\Phi_{(u',t')}(S).
\] The composition is induced by concatenation; identities are null programs. Well-definedness of composition on equivalence classes follows from determinism and the
concatenation property in Assumption~\ref{ass:deterministic}.
\end{definition}

\begin{proposition}
$\Prog$ is a category.
\end{proposition}
\begin{proof}
Associativity and identities follow from associativity of concatenation and the transition property in the assumption. \end{proof}

\subsection{Field-state evolution category $\Fld$ and semantics}

One may separate ``program syntax'' from ``evolution semantics'' by defining a category $\Fld$ of internal-state evolutions and a canonical functor $\llbracket\cdot\rrbracket:\Prog\to\Fld$. For brevity, we set $\Fld:=\Prog$ under Assumption~\ref{ass:deterministic}.

\subsection{Discretisation to mycelial networks}

We require an extraction map from internal state to observable mycelial network.

\begin{definition}[Extraction]
Let $\Pi:\mathcal{S}\to \Obj(\Myc)$ map an internal state $S$ to a mycelial object $\Pi(S)=(T,\sigma,\omega)$ via a chosen measurement/extraction pipeline (e.g.\ skeletonisation, graph inference, feature extraction).
\end{definition}

To make $\Pi$ functorial, we impose coherent tracking of identities along transitions.

\begin{assumption}[Coherent tracking]
For each morphism $[u,t]_S:S\to S'$ in $\Prog$ there is a canonical induced mycelial morphism
\[
\Pi([u,t]_S):\Pi(S)\to \Pi(S')
\]
such that identities and composition are preserved.
\end{assumption}

\begin{theorem}[Discretisation functor]
Under coherent tracking, $\Pi$ extends to a functor $\Pi:\Prog\to\Myc$.
\end{theorem}
\begin{proof}
By assumption, $\Pi$ maps identities to identities and respects composition of transitions, hence is a functor. \end{proof}

\subsection{Program-to-mycelium functor}

\begin{definition}[Program semantics functor]
Define
\[
\mathcal{F}_{\mathrm{prog}} := \Pi : \Prog\to \Myc,
\]
interpreting programs as mycelial transformations.
\end{definition}

\begin{remark}
This is the operational core: it is directly testable in controlled exposure experiments. The functor axiom asserts that executing a concatenated program yields the same mycelial transformation as composing transformations induced by the sub-programs.
\end{remark}

\subsection*{Biological Interpretation}

In controlled exposure experiments,
programs correspond to time-ordered sequences
of VOC pulses, humidity changes, nutrient injections,
or combined stimuli.
The functorial structure asserts that the fungal
response depends only on the induced state transition,
not on the syntactic description of the program.
This provides a structural criterion for reproducibility:
distinct exposure protocols that induce the same
state transition should yield identical extracted
mycelial morphisms.

\section{Local Lie Structure and Order Effects}

Sequential programs generally fail to commute: applying $P$ then $Q$
need not equal applying $Q$ then $P$. To formalise this structurally,
we introduce a local Lie-theoretic structure on the space of induced
state transformations.

\subsection*{9.1 Local Lie semigroup of state transformations}

Let $S$ denote the internal state space from Assumption~\ref{ass:deterministic}.
Consider the monoid $\mathrm{End}(S)$ of endomorphisms of $S$
induced by admissible programs.

We assume the following structural condition.

\textbf{Assumption 9.1 (Local Lie structure).}
There exists a neighbourhood $U$ of the identity transformation
in $\mathrm{End}(S)$ such that:

\begin{enumerate}
\item $U$ is closed under composition and inversion wherever defined;
\item composition is smooth in a local coordinate chart;
\item the tangent space at the identity defines a Lie algebra
      $\mathfrak{g}$.
\end{enumerate}

Under this assumption, sufficiently small programs correspond
to elements of a local Lie group, and their infinitesimal
generators lie in the Lie algebra $\mathfrak{g}$.

\subsection*{9.2 Infinitesimal generators}

For small-amplitude programs $P$ and $Q$ parameterised by $\varepsilon$,
assume their induced transformations admit exponential
representations of the form
\[
\Phi_P^{\varepsilon} = \exp(\varepsilon X_P),
\qquad
\Phi_Q^{\varepsilon} = \exp(\varepsilon X_Q),
\]
where $X_P, X_Q \in \mathfrak{g}$.

This formulation does not require $S$ itself to be a smooth manifold;
it requires only that the induced transformation monoid admits
a local Lie group structure near the identity.

\subsection*{9.3 Leading-order non-commutativity}

We assume $\mathfrak{g}$ is realised as a Lie algebra of
derivations of $S$ (or as a Lie subalgebra of $\mathrm{End}(S)$),
so that the bracket is given by the commutator.
The Lie bracket
\[
[X_P, X_Q] := X_P X_Q - X_Q X_P
\]
measures the failure of commutativity.

\textbf{Proposition 9.2 (Leading-order order asymmetry).}
Under Assumption~9.1, the composition satisfies
\[
\exp(\varepsilon X_Q)\exp(\varepsilon X_P)
=
\exp\!\left(
\varepsilon (X_P + X_Q)
+
\frac{\varepsilon^2}{2}[X_Q, X_P]
+ O(\varepsilon^3)
\right).
\]

Thus the commutator term governs the quadratic-order deviation
from commutativity.

\subsection*{9.4 Baker--Campbell--Hausdorff structure}

More generally, the composition law is given by the
Baker--Campbell--Hausdorff (BCH) expansion:
\[
\log\!\big(
\exp(\varepsilon X_Q)\exp(\varepsilon X_P)
\big)
=
\varepsilon (X_P + X_Q)
+
\frac{\varepsilon^2}{2}[X_Q, X_P]
+
\frac{\varepsilon^3}{12}
\big(
[X_Q,[X_Q,X_P]]
+
[X_P,[X_P,X_Q]]
\big)
+
\cdots .
\]

\textbf{Definition 9.3 (Effective mixture generator).}
Define the effective generator of the composed program
$P \star Q$ by
\[
X_{P \star Q}
:=
\log\!\big(
\exp(\varepsilon X_Q)\exp(\varepsilon X_P)
\big).
\]

Deviations from additivity are therefore governed by
Lie brackets and higher nested commutators.

\textbf{Remark 9.4.}
The above construction requires only a local Lie group structure
on induced state transformations and does not assume global
smoothness, linearity, or a specific mechanistic model for $S$.

\subsection*{Biological Interpretation}

The commutator $[X_P,X_Q]$ measures the degree to which
two small perturbations interact nonlinearly.
Experimentally, this predicts that the response to
a VOC mixture applied sequentially should deviate
from additivity by a term quadratic in perturbation
amplitude.
Higher nested commutators correspond to higher-order
mixture coupling effects.
This provides a concrete falsifiable prediction:
order asymmetry should scale quadratically
for sufficiently small exposures.

\section{Relation Between Static and Program Semantics}

The static functor $\mathcal{F}:\Env\to\Myc$ maps
environmental states and admissible environmental
transformations to mycelial states and network morphisms.
The program semantics functor
$\mathcal{F}_{\mathrm{prog}}:\Prog\to\Myc$
maps time-ordered exposure protocols to induced
mycelial transformations.

To relate these viewpoints, we introduce an
environmental evolution map analogous to the
internal-state evolution.

\begin{assumption}[Environmental evolution]
There exists a deterministic evolution map
\[
\Psi:\Env\times \{(u,t)\}\to \Env,
\qquad
(E,(u,t))\mapsto \Psi_{(u,t)}(E),
\]
respecting concatenation:
\[
\Psi_{(v,s)\star(u,t)}(E)
=
\Psi_{(v,s)}\big(\Psi_{(u,t)}(E)\big).
\]
\end{assumption}

\begin{definition}[Environment program category]\label{def:progenv}
Define a category $\Prog_{\Env}$ whose objects are environmental states
$E\in\Obj(\Env)$ and whose morphisms are those morphisms $f:E\to E'$
in $\Env$ that are realised by some admissible program $(u,t)$.
Composition and identities are inherited from $\Env$. We assume the class of realised morphisms is closed under identities and composition (equivalently, $\Prog_{\Env}$ is a subcategory of $\Env$).
\end{definition}

\begin{remark}
The map $\Psi$ determines the endpoints of realised morphisms.
The category $\Prog_{\Env}$ records the realised environmental
transformations at the level of $\Env$-morphisms, so the inclusion
$\mathcal{U}:\Prog_{\Env}\hookrightarrow \Env$ is canonical.
\end{remark}

\begin{assumption}[Environment-to-internal state selection]
There exists a (possibly partial) map $\iota:\Obj(\Env)\to\mathcal{S}$ assigning to an
environmental state $E$ an internal state $S_E:=\iota(E)$ compatible with $E$ under the
chosen experimental preparation and measurement pipeline.
\end{assumption}

We write $\mathcal{U}:\Prog_{\Env}\to\Env$ for the inclusion functor.
\begin{theorem}[Compatibility of static and operational semantics]
Suppose that environmental and internal evolutions are
compatible in the sense that for any program $(u,t)$
and environmental state $E$,
\[
\mathcal{F}\big(\Psi_{(u,t)}(E)\big)
=
\Pi\big(\Phi_{(u,t)}(S_E)\big),
\]
where $S_E:=\iota(E)$.
Then the diagram
\[
\begin{CD}
\Prog_{\Env} @>{r}>> \Prog \\
@V{\mathcal{U}}VV     @VV{\mathcal{F}_{\mathrm{prog}}}V \\
\Env @>>{\mathcal{F}}> \Myc
\end{CD}
\]
commutes up to the chosen identification.
Here $r:\Prog_{\Env}\to\Prog$ is a (generally non-canonical) functor.
On objects, $r(E):=S_E=\iota(E)$.
It selects, for each realised environmental morphism $f:E\to E'$, a representative program
class $r(f)=[u,t]_{S_E}:S_E\to S_{E'}$ realising the corresponding internal transition,
and does so coherently (preserving identities and composition).
Then, for each realised $f:E\to E'$, the chosen representative $r(f)$ satisfies
\[
\mathcal{F}(f)=\mathcal{F}_{\mathrm{prog}}(r(f))
\]
up to the chosen identification, and the square commutes up to that identification.
\end{theorem}

\section{A Minimal Worked Exposure Example}\label{sec:example}

We give a concrete, minimal instance of the framework in the setting of
a controlled two-exposure protocol. The purpose is not to model any
specific mechanism but to illustrate how the categorical objects,
morphisms, and non-commutativity claims translate into experimentally
measurable quantities.

\subsection*{Example 1 (Two VOC pulses on a linear electrode array)}

\textbf{Environment category.}
Let $G$ be a fixed substrate graph representing a one-dimensional
chamber region sampled by an ordered linear electrode array.
Let $\rho$ be held constant over the protocol (no nutrient injection),
and let $\phi:V(G)\to\mathbb{R}^2$ encode the concentrations of two
volatile compounds $A$ and $B$ (two channels in the chemical field).
Let $\chi$ encode fixed bounds on humidity and temperature (treated as
constant during the protocol).
Thus an environmental object is $E=(G,\rho,\phi,\chi)$.

Define two environmental morphisms
\[
f_A:E\to E_A, \qquad f_B:E\to E_B
\]
corresponding to two admissible exposure operations (e.g.\ short pulses)
that modify the chemical field component $\phi$ while leaving $G,\rho,\chi$
unchanged. In $\Env$, the composite $f_B\circ f_A$ represents applying
$A$ then $B$, while $f_A\circ f_B$ represents applying $B$ then $A$.

\textbf{Program category.}
Let $P$ and $Q$ be the corresponding programs in $\Prog$ implementing
the $A$-pulse and $B$-pulse with matched duration and amplitude scaling
parameter $\varepsilon>0$.
The concatenations $Q\star P$ and $P\star Q$ represent the two orderings.

\textbf{Extraction and observables.}
Let $\Pi:\mathcal{S}\to\Obj(\Myc)$ be an extraction pipeline that maps
the internal state to a mycelial network object $M=(T,\sigma,\omega)$.
In an electrophysiology experiment, $\omega$ may be taken as a vector of
time-series features per electrode site, for example:
\[
\omega(v) = (\text{spike rate},\ \text{median spike width},\ \text{median spike amplitude})
\]
computed on a fixed window after exposure, and $\sigma$ may be held fixed
(or interpreted as a slowly varying transport proxy on a longer timescale).
Thus the measured output of a program is the induced morphism in $\Myc$
between extracted states:
\[
\mathcal{F}_{\mathrm{prog}}(P):\Pi(S)\to\Pi(S_P),\qquad
\mathcal{F}_{\mathrm{prog}}(Q):\Pi(S)\to\Pi(S_Q),
\]
and similarly for the concatenations $Q\star P$ and $P\star Q$.

\textbf{Order effect as a commutator signature.}
Define an order-asymmetry observable by comparing extracted outcomes:
\[
\Delta_{\Pi}(P,Q;S)
:=
d_{\Myc}\!\left(
\Pi(\Phi_{Q\star P}(S)),
\Pi(\Phi_{P\star Q}(S))
\right),
\]
where $d_{\Myc}$ is any chosen distance on $\Myc$ compatible with the
extraction (e.g.\ a weighted sum of feature distances in $\omega$).
When $\Delta_{\Pi}$ is nonzero, the two exposures do not commute at the
level of extracted mycelial state.

\[
\begin{CD}
S @>{Q\star P}>> S_{QP} \\
@V{P\star Q}VV            @VV{\Pi}V \\
S_{PQ} @>>{\Pi}> \Pi(S_{PQ})
\end{CD}
\]

Comparison occurs at the level of
\[
\Pi(S_{QP}) \quad \text{versus} \quad \Pi(S_{PQ}).
\]

Under the local Lie structure of Section~9, and for sufficiently small
$\varepsilon$, the BCH expansion predicts that the leading-order order
effect scales quadratically:
\[
\Delta_{\Pi}(P,Q;S) = O(\varepsilon^2),
\]
with coefficient governed by the commutator term $[X_P,X_Q]$ (and higher
nested commutators at higher orders). In particular, if the commutator
vanishes in the explored region, then the two exposures commute up to
second order and the order-asymmetry becomes $O(\varepsilon^3)$.

\textbf{Interpretation.}
This example shows how the framework turns the qualitative statement
``order matters'' into a falsifiable structural prediction: for small
perturbations, order asymmetry should grow approximately quadratically
with exposure amplitude (or duration scaling) and should vanish when the
two exposures commute in the induced local Lie algebra.

\section{Scale Separation and Structural Levels}

The framework operates simultaneously at multiple structural scales,
which should be clearly distinguished.

\begin{itemize}
\item The category $\Env$ encodes spatially structured environmental
states at chamber or substrate scale (minutes to hours, centimetres to
tens of centimetres).

\item The category $\Prog$ encodes temporally ordered exposure
protocols at experimental-control scale (seconds to hours).

\item The category $\Myc$ encodes observable network organisation
and electrophysiological state at organismal scale
(hyphal networks, electrode arrays, transport structure).

\item The local Lie structure of Section~9 is perturbative:
it describes infinitesimal neighbourhoods of the identity
transformation in the space of induced state transitions,
and therefore applies to sufficiently small exposures.
\end{itemize}

These levels are related functorially but are not collapsed
into a single dynamical model. The theory does not assume
that spatial, temporal, and network scales coincide;
rather, it formalises how transformations at one level
induce transformations at another. This separation of scales
allows experimental protocols, environmental geometry,
and organismal adaptation to be analysed structurally
without committing to a specific mechanistic equation set.

\section{Conclusion}

We have provided a rigorous, equation-free category-theoretic foundation for fungal organisation. The main constructs are: a category $\Env$ of structured environments; a category $\Myc$ of mycelial network states; a fungal functor $\mathcal{F}:\Env\to\Myc$; a program semantics functor $\mathcal{F}_{\mathrm{prog}}:\Prog\to\Myc$ capturing operational perturbation--response; natural transformations encoding species/strain variability; an adjunction capturing ecological feedback; categorical limits/colimits corresponding to network fusion and constraint interactions; and a Lie/BCH structure quantifying non-commutativity and mixture coupling near identity programs. The theory yields falsifiable predictions: order asymmetry should scale quadratically for small perturbations with coefficient given by a commutator, and increasing mixture complexity should manifest as growth of nested commutator terms.

\end{document}